\documentclass[pra,reprint,amssymb,superscriptaddress,longbibliography]{revtex4-1}

\usepackage{amsthm}
\usepackage{amsmath}
\usepackage{graphicx}

\usepackage{physics}

\usepackage{hyperref}
\hypersetup{colorlinks=true,linkcolor=blue,citecolor=blue,urlcolor=blue}

\usepackage{tabularx}
\usepackage{multirow}

\usepackage[math]{cellspace}
\newcommand*\diff{\mathop{}\!\mathrm{d}}
\newtheorem{theorem}{Theorem}

\begin{document}

\title{Overarching framework between Gaussian quantum discord and Gaussian quantum illumination}

\author{Mark~Bradshaw}
\affiliation{Centre for Quantum Computation and Communication Technology, Department of Quantum Science,\\ Research School of Physics and Engineering, Australian National University, Canberra ACT 2601, Australia}
\author{Syed~M.~Assad}
\affiliation{Centre for Quantum Computation and Communication Technology, Department of Quantum Science,\\ Research School of Physics and Engineering, Australian National University, Canberra ACT 2601, Australia}
\author{Jing~Yan~Haw}
\affiliation{Centre for Quantum Computation and Communication Technology, Department of Quantum Science,\\ Research School of Physics and Engineering, Australian National University, Canberra ACT 2601, Australia}
\author{Si-Hui~Tan}
\affiliation{Singapore University of Technology and Design, 8 Somapah Road, Republic of Singapore}
\author{Ping~Koy~Lam}
\affiliation{Centre for Quantum Computation and Communication Technology, Department of Quantum Science,\\ Research School of Physics and Engineering, Australian National University, Canberra ACT 2601, Australia}
\author{Mile~Gu}
\affiliation{School of Physical and Mathematical Sciences, Nanyang Technological University, Singapore 639673, Republic of Singapore}
\affiliation{Complexity Institute, Nanyang Technological University, Singapore 639673, Republic of Singapore}
\affiliation{Centre for Quantum Technologies, National University of Singapore, 3 Science Drive 2, Singapore, Republic of Singapore}

\date{\today}

\begin{abstract}
We cast the problem of illuminating an object in a noisy environment
into a communication protocol. A probe is sent into the environment,
and the presence or absence of the object constitutes a signal encoded
on the probe. The probe is then measured to decode the signal. We calculate the Holevo information and bounds to the accessible information between the encoded and received signal with two different
Gaussian probes---an Einstein-Podolsky-Rosen (EPR) state and a coherent state.
We also evaluate the Gaussian discord consumed during the encoding
process with the EPR probe. We find that the Holevo quantum
advantage, defined as the difference between the Holevo information
obtained from the EPR and coherent state probes, is approximately equal
to the discord consumed. These quantities become exact in the typical
illumination regime of low object reflectivity and low probe energy.
Hence we show that discord is the resource responsible for the quantum
advantage in Gaussian quantum illumination.
\end{abstract}

\maketitle

\section{Introduction}

Quantum illumination is a simple target-detection scheme, first proposed by Lloyd for photonic qubits \cite{Lloyd1463}. It harnesses entanglement in a quantum state of light to better infer the presence or absence of a weakly reflecting object flooded by white noise. The protocol distinguished itself in displaying quantum advantage, even in regimes so noisy that no entanglement survives. It presented a remarkable deviation from the conventional view that quantum technologies are fragile, displaying advantage only in carefully engineered environments which ensure little or no loss of entanglement. Since its original inception, quantum illumination has gained significant scientific interest. Many variants have been proposed, including some that make use of Gaussian states in the continuous-variable regime~\cite{PhysRevLett.101.253601,SHGSC16,PhysRevA.80.052310} and inspiring a number of different experimental realizations \cite{PhysRevLett.110.153603, PhysRevLett.111.010501,PhysRevLett.114.110506, PhysRevLett.114.080503}. The phenomenon has also seen applications outside metrology, where quantum illumination has been harnessed to provide security against passive eavesdropping in the setting of secure communication~\cite{PhysRevA.80.022320}.

Quantum illumination challenges the conventional view that entanglement alone can explain all quantum advantage. It joins a particularly surprising class of protocols that appear to thrive in noisy, possibly entanglement-breaking environments~\cite{PhysRevLett.81.5672, PhysRevLett.101.200501}. What other quantum resources, then, could help us better understand its noisy resilience?  Quantum discord \cite{0305-4470-34-35-315, PhysRevLett.88.017901, RevModPhys.84.1655}, which quantifies correlations beyond entanglement, is considered a likely candidate. Unlike entanglement, discord is far more robust and can also survive in highly noisy conditions \cite{PhysRevLett.100.050502}. In fact, Weedbrook {\it et al.}\ have shown such a relation for discrete variables~\cite{weedbrook2016discord}. Specifically, they showed that the performance advantage of quantum illumination---in terms of extra accessible information about whether an object is present---can be directly related to the amount of discord in the illumination protocol that survives after being subjected to entanglement-breaking noise. Does a similar relationship hold for continuous variables?

The aim of this work is to answer that question. We extend the framework relating discord and illumination to the continuous variable regime. This involves understanding how these relations generalize when a number of conditions specific to the discrete scenario no longer hold. The paper is organized as follows. In Sec.\ \ref{sec_framework} we describe the illumination protocol and the quantifiers of performance. In Sec.\ \ref{sec_discord} we describe discord and how it relates to quantum illumination. In Sec.\ \ref{sec_results} we present and discuss our results, demonstrating that there is a general relationship between discord and the quantum advantage of illumination in the continuous-variable regime.

\section{The illumination framework}
\label{sec_framework}

\begin{figure}[h]
\begin{center}
\includegraphics[width=7.4cm]{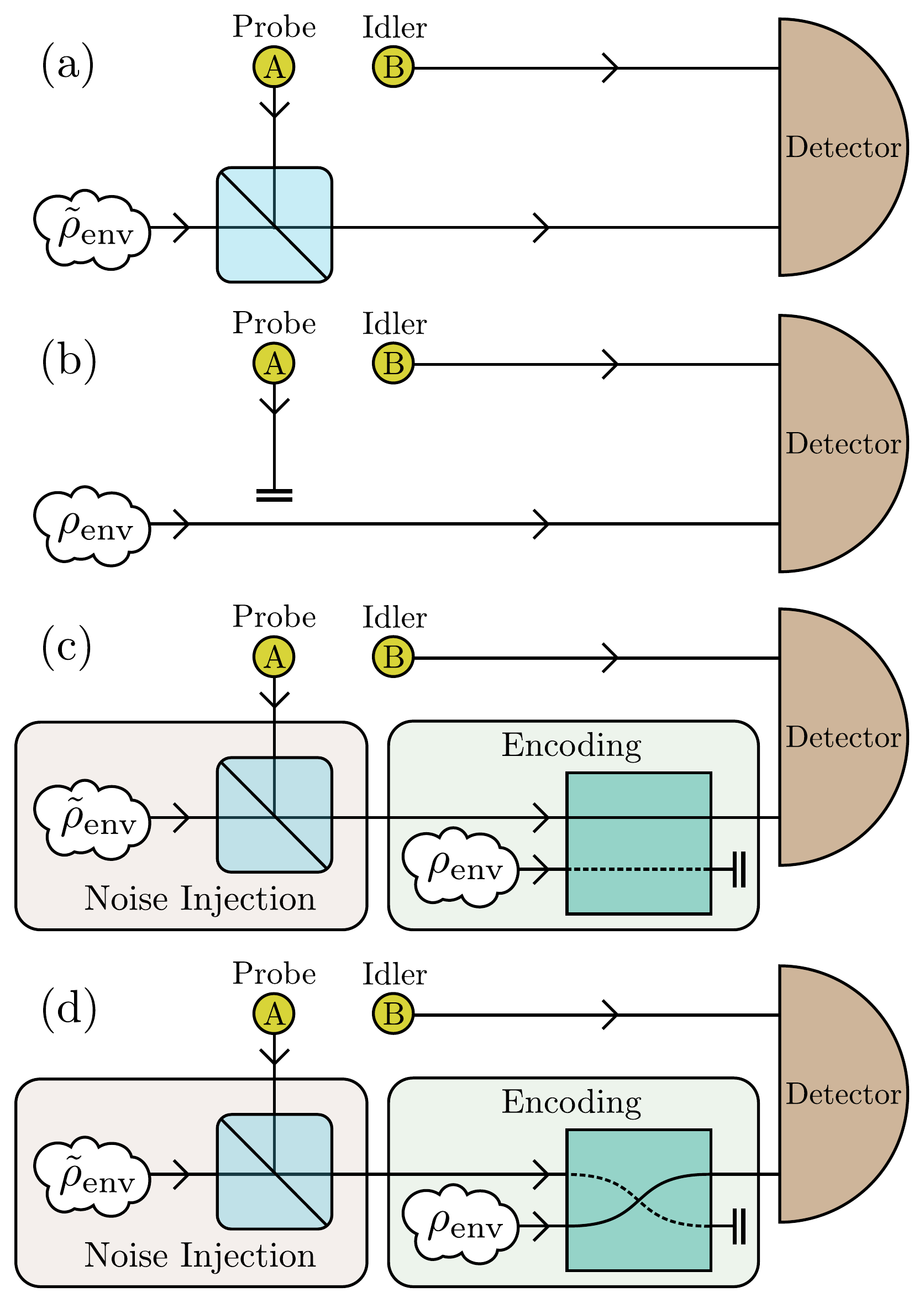}
\caption{ \textbf{Diagram of illumination setup.} (a) With probability $p_0$, there is an object located in a noisy environment. The object is partially reflective (modeled as a beam splitter with reflectivity $\epsilon$). A probe is sent towards the object. The probe is mixed with the noisy environment, and reflected to the detector. (b) With probability $p_1$, an object is not present, in which case there is nothing to reflect the probe to the detector. Hence, only noise is detected. (c) An equivalent description of illumination whereby first noise is injected. Then, encoding is performed on the probe, whereby with a probability $p_0$, an identity operation is performed on the probe (after noise injection) and the environment noise, and (d) with probability $p_1$, a swap operation is performed on the probe and environment. In quantum illumination, we also have an idler initially entangled with the probe which is used to perform a joint measurement. Single-mode illumination is when there is no idler. $\rho_{\rm env}$ is the noisy environment and $\tilde{\rho}_{\rm env}$ is the environment with the mean photon number scaled by $1/(1-\epsilon)$. }
\label{fig_swapdiagram}
\end{center}
\end{figure}

\subsection{Setup}
The illumination framework is described as follows: Bob wishes to determine whether an object is located in a noisy environment. He sends a quantum state, referred to as the probe, to the location. If an object is present, part of the probe will be reflected back to Bob, along with some background noise. If the object is not present, Bob receives only the background noise. Bob may have another state called the idler, which was initially correlated with the probe. 

If the probe and idler are quantum correlated (have a non-zero quantum discord), the scheme is called \emph{quantum illumination}. If there is no idler, it is called \emph{single-mode illumination}.

A diagram of illumination is shown in Figs.~\ref{fig_swapdiagram}(a) and ~\ref{fig_swapdiagram}(b). Bob performs a joint measurement on the idler and returning probe, and uses the results of the measurement to determine whether an object was present. For brevity in notation in the rest of the paper, modes A and B will label the probe and idler parts of the state, respectively.

We are interested in quantum illumination in the continuous-variable setting, where the probe and idler are Gaussian states. For single-mode illumination, Bob uses a coherent state $\rho_\alpha$, where $\alpha$ is its amplitude. For quantum illumination,  Bob uses an Einstein-Podolsky-Rosen (EPR) state described by $\rho_{\rm EPR}=\ket{\psi_{\rm EPR}}\bra{\psi_{\rm EPR}}
$, where
\begin{align}\label{eqn_epr}
\ket{\psi_{\rm EPR}}&=\sqrt{1-\lambda^2}\sum_{n=0}^\infty(-\lambda)^n\ket{n}_\mathrm{A}\ket{n}_\mathrm{B} \ .
\end{align}
where $\lambda=\tanh(r)$, and $r$ is the squeezing parameter.

Illumination can also be recast as a communication protocol. Let us suppose that Alice is in control of the object and she would like to communicate with Bob. She can do so by encoding a binary alphabet via the control of the object, such as in the Morse code. The message she sends to Bob can be described by realizations of a random variable $X$, where if $X=0$, Alice places the object in the noisy environment, and if $X=1$, Alice removes the object. Let $p_x$ be the prior probability that $X=x$, and let $p_0=p_1$, i.e., let both hypotheses be equally likely to occur. Let $\rho^{(x)}$ denote the state received by Bob when $X=x$. Noise is injected into the probe state before Alice encodes the value of $X$. This is shown diagrammatically in Figs.~\ref{fig_swapdiagram}(c) and \ref{fig_swapdiagram}(d). We model the object as a beam splitter with reflectivity $\epsilon$. The environment-noise state $\rho_{\rm env}$ is a thermal state with mean photon number $\bar n_{\rm env}$, where $\rho_{\rm env}(\bar{n})=\sum_{n=0}^\infty\frac{\bar{n}^n}{(\bar{n}+1)^{n+1}}\ket{n}\bra{n}$. When the object is present, the environment noise is multiplied by a factor of $1/(1-\epsilon)$ such that the mean number of noise photons arriving at the detector is the same as when the object is absent. This approach has been adopted by \cite{PhysRevLett.101.253601} to avoid a ``shadowing effect''---so that the object is not detected by a reduction in the number of noise photons arriving at the detector. The typical illumination scenario that has greatest quantum advantage is for the regime of low object reflectivity and high noise, i.e., $\epsilon\ll 1$ and $\bar n \ll \bar n_{\rm env}$, where $\bar n$ is the mean photon number of the probe. We term this as the intense white-noise limit.

Consider Figs.\ \ref{fig_swapdiagram}(c) and \ref{fig_swapdiagram}(d). After the noise injection, the entanglement is reduced or lost all together, before any information is encoded within the probe. In fact, for all the settings studied in Sec.\ \ref{sec_results}, the entanglement after noise injection is strictly zero. Nevertheless we see a quantum advantage. Thus, quantum entanglement itself does not give a complete picture on why illumination thrives in such noise. Our goal here is to see if discord will give us additional insight.

In the next section, we will use the communication formalism to study the amount of information that Alice can communicate to Bob under different settings. This provides a measure for assessing the performance of illumination under these settings.

\subsection{Quantifiers of performance}

We consider two quantifiers of performance of illumination: the accessible information and Holevo information.

Let $\mathcal{M} = {\{E_k\}}$ be a positive operator-valued measure (POVM) that mathematically represents a measurement. The POVM elements $E_k$ are non-negative, self-adjoint operators satisfying $\sum_k E_k = 1$, where the subscript $k$ labels the outcome of the measurement. The probability of the measurement outcome $k$ on a state $\rho^{(x)}$ is then given by $q_k^{(x)}=\Tr (\rho^{(x)} E_k)$. Let this be governed by random variable $K_{\mathcal{M}}$.
In the communication setting described in the last section, the amount of information obtained by Bob after measurement of the state $\rho^{(x)}$ is given by the mutual information,
\begin{equation}
\label{eqn_mutual_int}
I_{\rm mut}(X,K_{\mathcal{M}})
 = \sum_k \sum_{x=0}^1 p_x q^{(x)}_k \log_2 (\frac{q^{(x)}_k}{q_k}),
\end{equation}
where $q_k=\sum_{x=0}^1 p_x q^{(x)}_k$. The accessible information is the maximization of the mutual information over all POVMs,
\begin{equation}
\label{eqn_accessible}
A\left(\rho^{(0)},\rho^{(1)}\right) = \max_{\mathcal{M}}I_{\rm mut}\left(X,K_{\mathcal{M}}\right).
\end{equation}
The accessible information quantifies Bob's knowledge when each $\rho^{(x)}$ from $N$ trials is measured separately using an optimal POVM. In the context of communication, illumination can be regarded as classical information exchange over a noisy channel. By the Shannon's noisy-channel coding theorem~\cite{6773024}, Alice and Bob communicate at a rate equal to the accessible information in the limit of infinite message size $N$.

There is no known general method for calculating the accessible information exactly. Here we will make use of the upper and lower bounds found by Fuchs and Caves~\cite{PhysRevLett.73.3047}. The lower bound, hereby referred to as the Fuch's lower bound, is
\begin{multline}
I_{\rm lower}=\Tr\Big\{p_0\rho^{(0)}\log_2\left[ \mathcal{L}_{\bar\rho}(\rho^{(0)}) \right]\\
 + p_1\rho^{(1)}\log_2\left[ \mathcal{L}_{\bar\rho}(\rho^{(1)})\right] \Big\}
\label{eqn_fuchslower}
\end{multline}
where  $\mathcal{L}$ is the lowering superoperator given by
\begin{multline}
\mathcal{L}_{\bar\rho}(\Delta)= \sum_{\{j,k|\lambda_j+\lambda_k\neq 0\}} \bigg[ \frac{2}{\lambda_j(p_1)+\lambda_k(p_1)} \\ \times \bra{\psi_j(p_1)}\Delta\ket{\psi_j(p_1)}
 \ket{\psi_j(p_1)}\bra{\psi_k(p_1)}\bigg] \ ,
\end{multline}
and where $\Delta=\rho^{(1)}-\rho^{(0)}$. $\lambda_i(p_1)$ and $\ket{\psi_i(p_1)}$ are the eigenvalues and eigenvectors of $\bar\rho=(1-p_1)\rho^{(0)}+p_1\rho^{(1)}$.
The Fuchs upper bound $I_{\rm upper}$ is found by numerically solving the differential equation,
\begin{multline}
\label{eqn_fuchsupper}
\frac{\diff^2 I_{\rm upper}(p_1)}{\diff p_1^2}=
 \sum_{\{ j,k|\lambda_j+\lambda_k\neq 0 \}} \bigg[ -\frac{2}{\lambda_j(p_1)+\lambda_k(p_1)} \\
 \times |\bra{\psi_j(p_1)}\Delta\ket{\psi_k(p_1)}|^2 \bigg]
\end{multline}
subject to
\begin{equation}
I_{\rm upper}(0)=I_{\rm upper}(1)=0.
\end{equation}

The other figure of merit we consider is the Holevo information~\cite{holevo1973bounds}. It is given by
\begin{multline}
\label{eqn_holevo}
\chi (\rho^{(0)}, \rho^{(1)}) = S\left(\sum_{x=0}^1 p_x \rho^{(x)}\right) - \sum_{x=0}^1 p_x S(\rho^{(x)})
\end{multline}
where $S(\rho)$ is the von Neumann entropy of the quantum state $\rho$. The Holevo information is the maximum communication rate Bob can obtain, provided he stores all of the $N$ states and then performs a joint measurement upon all of the states. From the Holevo-Schumacher-Westmoreland theorem~\cite{Schumacher97sendingclassical, 651037}, this information rate is obtainable when $N\rightarrow \infty$.

\subsection{Three cases of illumination and quantum advantage}
\label{sec_three_cases}
Three cases, together with three pairs of accessible information and Holevo information, are relevant for our assessment of the illumination scheme [Fig.\ \ref{fig_swapdiagram}(a)] in the communication framework. They are as follows:\\

\textit{Case 1}.\ Quantum illumination with joint measurement: $A_q$ and $\chi_q$ are the accessible information and Holevo information, respectively, for Bob when two-mode EPR states are used as probes and idlers for illumination. Any arbitrary joint measurement over the two modes is allowed.\\

\textit{Case 2}.\ Quantum illumination with local measurements: $A_c$ and $\chi_c$ are the average accessible information and Holevo information for Bob with EPR state as the probe and idler, under the restriction that Bob must perform the optimal Gaussian local measurement on mode B, followed by an arbitrary local measurement on mode A. The measurement on mode B is optimal in the sense that it maximizes the amount of accessible information or Holevo information Bob receives. In this case, Bob only takes advantage of the classical correlations of the EPR state. This enables a direct comparison to case 1, when both quantum and classical correlations are utilized. \\

\textit{Case 3}.\ Single-mode illumination: $A_s$ and $\chi_s$ are the accessible information and Holevo information, respectively, when Bob uses a single-mode coherent state with a fixed amplitude $\alpha$ as the illumination probe. \\

The quantum advantage is defined as the difference between the performance of quantum illumination and single-mode illumination protocol. The protocols are compared for scenarios where the probe states have coinciding energy. This constraint allows for fair comparison, as it is always possible to detect the presence of an object with any fixed accuracy by using a sufficiently energetic probe. The quantum advantage in terms of accessible information is $A_q-A_s$ and the Holevo information quantum advantage is $\chi_q-\chi_s$, where each information quantity is evaluated over the probe with mean photon number $\bar n$. As we shall show in this paper, these quantum advantages can be linked to the discord consumed in the illumination protocol.

\section{Discord and Quantum Illumination}
\label{sec_discord}
Quantum discord is a measure of the nonclassical correlations between two quantum states. It arises from the difference between quantum analogs of two distinct definitions of the classical mutual information~\cite{0305-4470-34-35-315, PhysRevLett.88.017901}:
\begin{align}
I({\rm A}:{\rm B})&=S({\rm A})+S({\rm B})-S({\rm AB}) \\
J({\rm A}|{\rm B})&=S({\rm A})-\min_{\{\Pi_b\}}\sum p_b S({\rm A}|b)
\end{align}
where $\Pi_b$ is the POVM element corresponding to the outcome $b$, $p_b$ is the probability of that outcome, and $S({\rm A}|b)$ is the entropy of the state conditioned on the outcome $b$. The discord is then
\begin{align}
\label{eq_discord}
\delta({\rm A}|{\rm B}) &= I({\rm A}:{\rm B})-J({\rm A}|{\rm B}) \nonumber \\
&= S({\rm B}) - S({\rm AB}) + \min_{\{\Pi_b\}}\sum p_b S({\rm A}|b) \ ,
\end{align}
where the minimization is done over all possible POVMs on mode B. In the special case that the domain of this minimization is restricted to Gaussian measurements, the discord is known as the Gaussian discord \cite{PhysRevLett.105.020503, PhysRevLett.105.030501}. It was recently shown that for a large class of Gaussian states, Gaussian quantum discord is equal to quantum discord~\cite{PhysRevLett.113.140405}. Henceforth, we denote the Gaussian discord $\delta^{\rm G}({\rm A}|{\rm B})$ with a superscript G.

We now consider the evolution of the discord when quantum illumination is described by Figs.\ \ref{fig_swapdiagram}(c) and \ref{fig_swapdiagram}(d). After the noise-injection step, Alice is left with state $\rho$ with which she can encode information to send to Bob. We note that this state may have no entanglement due to the noise injection~\cite{PhysRevLett.101.253601}. Alice encodes the value of $X$ on the state by performing the operation $O_x$ on $\rho$, resulting in a state $\rho^{(x)}=O_x(\rho)$ with discord $\delta^{(x)}({\rm A}|{\rm B})$.

Let us decompose the discord of $\rho$, $\delta({\rm A}|{\rm B})$ into three components:
\begin{equation}
\label{eq_discord_decomp}
\delta({\rm A}|{\rm B}) = \delta_{\rm loss} + \bar\delta({\rm A}|{\rm B}) + \delta_{\rm con}({\rm A}|{\rm B})
\end{equation}
The first component $\delta_{\rm loss}$ is the amount of discord lost to the environment during the encoding process. This can be evaluated by first defining
\begin{equation}
\delta_{\rm loss}^{(x)} = \delta({\rm A}|{\rm B}) - \delta^{(x)}({\rm A}|{\rm B})
\end{equation}
as the loss of discord for each possible value of $x$ that Alice can encode, and then taking the weighted average over the probability of encoding that $x$. This results in 
\begin{equation}
\delta_{\rm loss} = \sum_x p_x \delta_{\rm loss}^{(x)}
\end{equation}
The second component $\bar\delta({\rm A}|{\rm B})$ is the discord of $\bar\rho=p_0\rho^{(0)}+p_1\rho^{(1)}$, the state after encoding. This is the state seen by Bob who is oblivious to the value of $X$.

We term the remaining component the {\it consumed discord} $\delta_{\rm con}({\rm A}|{\rm B})$, and represents the discord in $\rho$ that remain unaccounted for. In prior literature, it was proposed to capture the amount of discord consumed to encode the value of $X$ on the state $\rho$~\cite{weedbrook2016discord}. For the special case where encodings were unitary, such that $\delta_{\rm loss}^{(x)} = 0$, $\delta_{\rm con}({\rm A}|{\rm B})$ was related to the advantage of using coherent interactions~\cite{Gu:2012aa}. It is also interesting to note that $\delta_{\rm con}({\rm A}|{\rm B})$ also coincides with the the extra discord Bob sees between $A$ and $B$, should he learn the value of $X$.

In quantum illumination, when $X=0$, Alice performs an identity operation, and thus $\delta^{(0)}({\rm A}|{\rm B})=\delta({\rm A}|{\rm B})$ and $\delta_{\rm loss}^{(0)} = 0$. When $X=1$, Alice performs a swap operation between mode A of $\rho$ with the environment noise, destroying all correlations between the two modes. All discord is lost and $\delta_{\rm loss}^{(1)}=\delta({\rm A}|{\rm B})$. Putting this together, the average discord loss is thus $\delta_{\rm loss}=p_1\delta({\rm A}|{\rm B})$. Hence the consumed discord for quantum illumination is
\begin{equation}
\label{eqn_denc}
\delta_{\rm con}(A|B)=p_0\delta^{(0)}({\rm A}|{\rm B})-\bar\delta({\rm A}|{\rm B}).
\end{equation}

\section{Method and Results}
\label{sec_results}

In Sec.\ \ref{sec_results_analytic}, we first derive a general result that if certain conditions are fulfilled, the discord consumed is equal to the Holevo information quantum advantage. In Sec.\ \ref{sec_results_info}, we numerically calculate the illumination information quantities. In Sec.\ \ref{sec_results_discord}, we numerically evaluate the consumed discord and compare it to the quantum advantages. Our main result is that for continuous variable quantum illumination, the consumed discord is approximately equal to the Holevo information quantum advantage.

\begin{figure*}
\includegraphics[width=15.5cm]{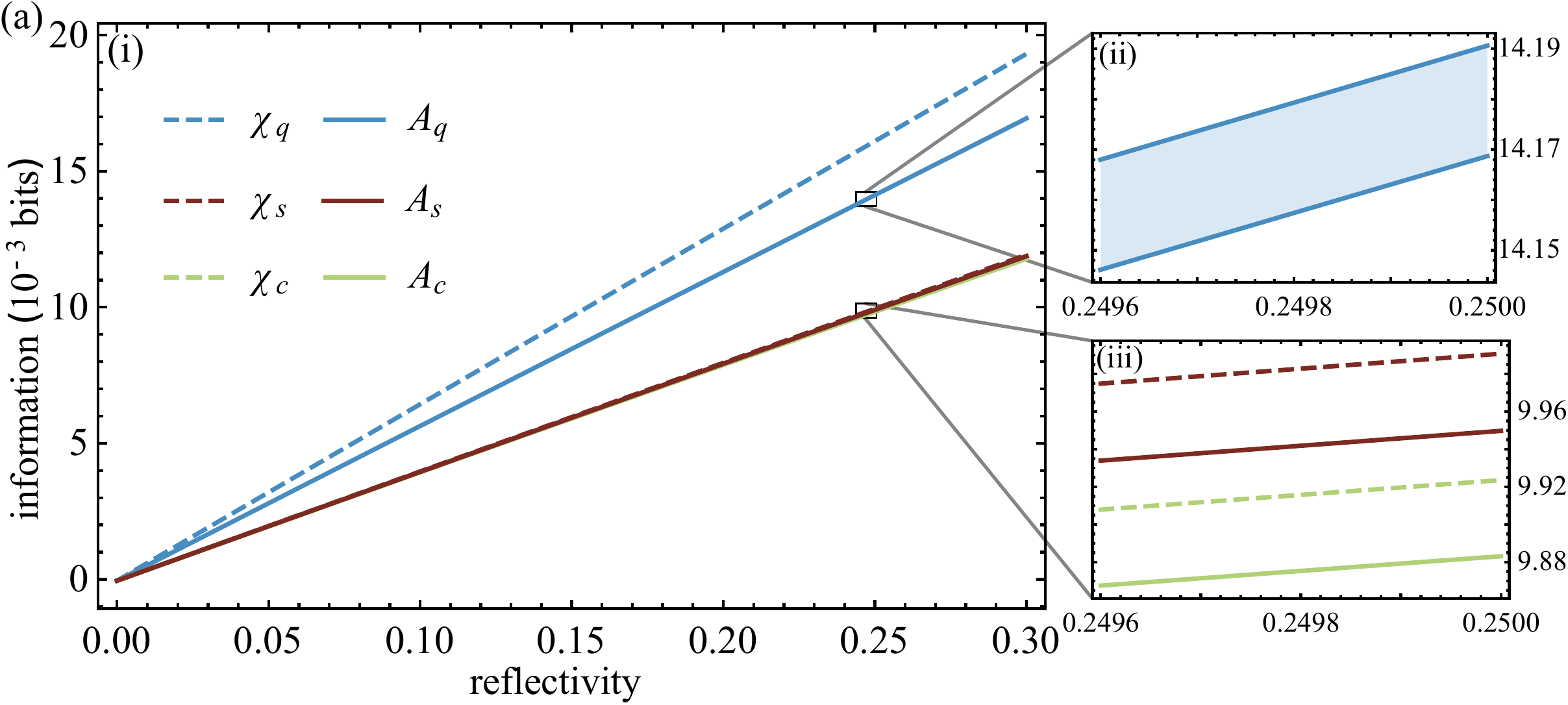}
\includegraphics[width=15.5cm]{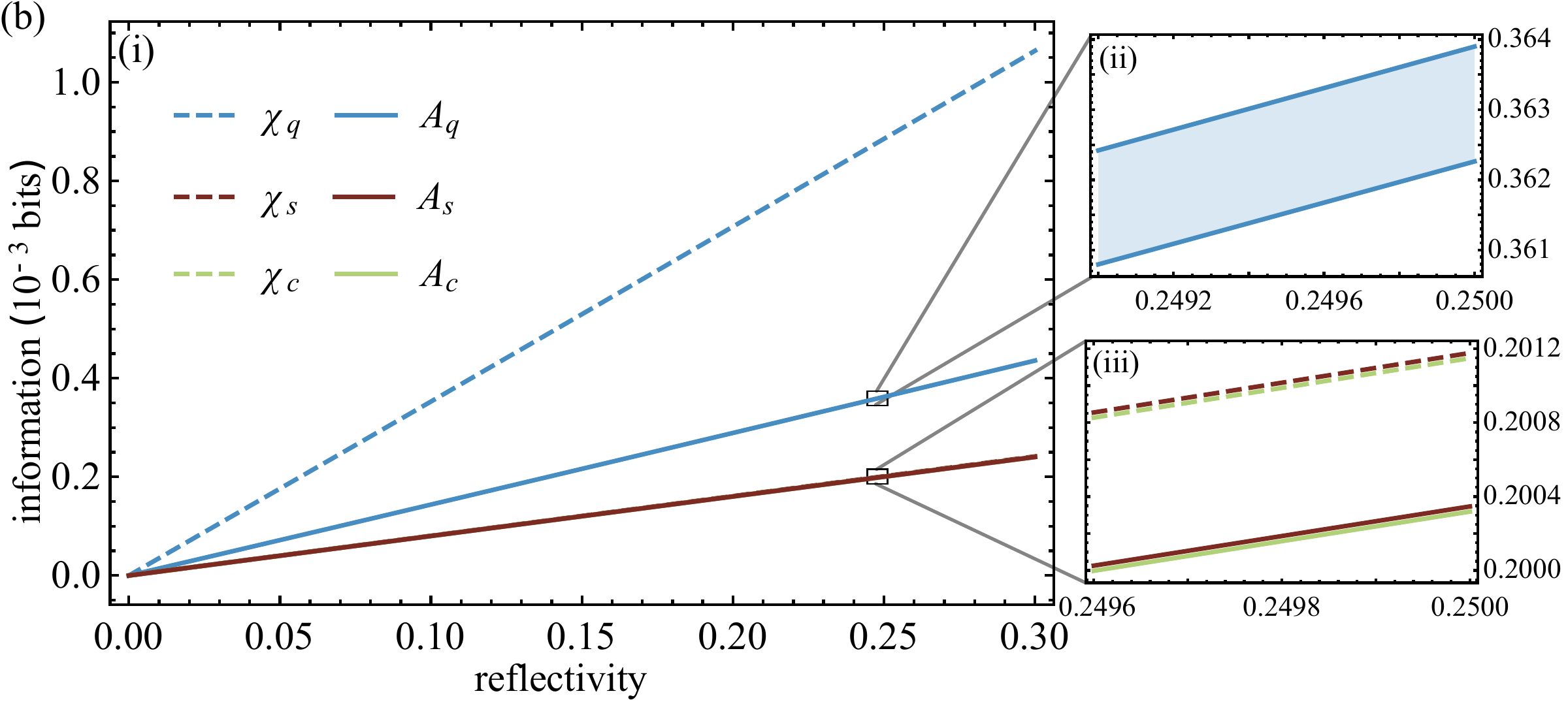}
\caption{Information vs object reflectivity $\epsilon$ when probe has mean photon number (a) $0.5 $ and (b) $0.01$. The environment noise has mean photon number $4$. Each plot has two insets showing zoomed portions. Insets (ii) show the upper and lower bounds for $A_q$, with the true value lying somewhere in the shaded region. Insets (iii) show that $\chi_s$, $\chi_c$, $A_s$, and $A_c$ differ slightly, despite appearing as a single line in the main plot. }
\label{fig_plot01}
\end{figure*}

\subsection{Analytic Result}
\label{sec_results_analytic}

We prove the following theorem.

\begin{theorem}
\label{theorem1}
Let $\rho_{{\rm AB}}^{(0)}$ and $\rho_{{\rm AB}}^{(1)}$ be two arbitrary two mode states. If the following conditions are met:
\begin{enumerate}
\item mode B is the same for both states, i.e., $\rho_{\rm B}^{(0)}=\rho_B^{(1)}$ where $\rho_{\rm B}^{(x)}=\Tr_A(\rho_{{\rm AB}}^{(x)})$ and where $\Tr_{\rm A}$ denotes the partial trace over subsystem A;
\item $\rho_{{\rm AB}}^{(1)}$ is a product state, i.e., $\rho_{{\rm AB}}^{(1)}=\rho_A^{(1)}\otimes\rho_B^{(1)}$; and
\item the Holevo information of local measurement $\chi_c$, the discord of $\bar \rho_{{\rm AB}}=p_0\rho_{{\rm AB}}^{(0)}+p_1\rho_{{\rm AB}}^{(1)}$, and the discord of $\rho_{{\rm AB}}^{(0)}$ are achieved by the same measurement,
\end{enumerate}
then $\delta_{\rm con}({\rm A}|{\rm B})=\chi_q-\chi_c$, where
\begin{align*}
\chi_q&=\chi(\rho_{{\rm AB}}^{(0)},\rho_{{\rm AB}}^{(1)})\\
\chi_c&=\max_{\{\Pi_b\}}\sum_b p_b \chi(\rho_{{\rm A}|b}^{(0)},\rho_{{\rm A}|b}^{(1)}),
\end{align*}
where $p_b$ is the probability of measuring outcome $\Pi_b$ on subsystem B, and $\rho_{{\rm A}|b}^{(x)}$ are the states of subsystem A conditioned on that outcome.
\end{theorem}

\begin{proof}

Let $\{\Pi_b\}$ be the measurement in condition 3 that simultaneously optimizes $\chi_c$, as well as the discord of states $\bar \rho_{{\rm AB}}$ and $\rho_{{\rm AB}}^{(0)}$. The measurement outcome probability is
\begin{equation*}
p_b=\Tr((\Pi_b\otimes I)\rho_{{\rm AB}}^{(0)})=\Tr((\Pi_b\otimes I)\rho_{{\rm AB}}^{(1)}),
\end{equation*}
where we have used condition 1.
The resulting conditional states are
\begin{equation*}
\rho^{(x)}_{{\rm A}|b}=\frac{\Tr_{\rm B}(\Pi_b\rho_{{\rm AB}}^{(x)})}{p_b}.
\end{equation*}

Our goal is to prove $\delta_{\rm con}({\rm A}|{\rm B})=\chi_q-\chi_c$. Because of condition 2, $\delta^{(1)}({\rm A}|{\rm B})=0$, as so the consumed discord is
\begin{align*}
\delta_{\rm con}({\rm A}|{\rm B}) &= p_0\delta^{(0)}({\rm A}|{\rm B})-\bar\delta({\rm A}|{\rm B}) \\
&= p_0(S(\rho_{\rm B}^{(0)})-S(\rho_{{\rm AB}}^{(0)})+\sum_b p_bS(\rho^{(0)}_{{\rm A}|b})) \\
&-S(\bar\rho_{\rm B})+S(\bar\rho_{{\rm AB}})-\sum_b p_b S(\bar\rho_{{\rm A}|b}).
\end{align*}
We also have that:
\begin{align*}
\chi_q-\chi_c&=S(\bar\rho_{{\rm AB}})-p_0S(\rho_{{\rm AB}}^{(0)})-p_1S(\rho_{{\rm AB}}^{(1)}) \\
& +\sum_bp_b (-S(\bar\rho_{{\rm A}|b})+p_0S(\rho_{{\rm A}|b}^{(0)})+p_1S(\rho_{{\rm A}|b}^{(1)})).
\end{align*}
This leads to
\begin{align*}
&\delta_{\rm con}({\rm A}|{\rm B})-(\chi_q-\chi_c)\\
&=p_0 S(\rho_{\rm B}^{(0)}) -S(\bar\rho_{\rm B})+p_1S(\rho_{{\rm AB}}^{(1)}) -\sum_bp_bp_1S(\rho_{{\rm A}|b}^{(1)}).
\end{align*}
From condition 1, we have that $\rho_{\rm B}^{(0)}=\rho_{\rm B}^{(1)}=\bar\rho_{\rm B}$. From condition 2, $\rho_{{\rm AB}}^{(1)}$ is a product state, so $S(\rho_{{\rm AB}}^{(1)})=S(\rho_{\rm A}^{(1)})+S(\rho_{{\rm B}}^{(1)})$ and $\rho^{(1)}_{{\rm A}|b}=\rho^{(1)}_{\rm A}$. So this becomes
\begin{align*}
&\delta_{\rm con}({\rm A}|{\rm B})-(\chi_q-\chi_c)\\
&=S(\rho_{\rm B}^{(0)})(p_0-1+p_1)+S(\rho_{\rm A}^{(1)})(p_1-p_1)\\
&=0.
\end{align*}

\end{proof}

In continuous-variable quantum illumination, condition 1 is satisfied since the idler is not interacting with the illumination object. Condition 2 is met by the fact that the swap operation decorrelates mode A and mode B. By restricting ourselves to Gaussian quantum discord, together with the assumption that a Gaussian heterodyne measurement is the optimal measurement for the quantities in condition 3, we have $\delta^{\rm G}_{\rm con}({\rm A}|{\rm B})=\chi_q-\chi_c$. This assumption is justified by numerical results in the next sections.

\subsection{Accessible information and Holevo information calculations}
\label{sec_results_info}

The accessible information and Holevo information quantities $A_q$, $\chi_q$, $A_c$, $\chi_c$, $A_s$ and $\chi_s$ were calculated numerically for typical settings of quantum illumination. Due to finite computational resources, the states must be approximated to a Hilbert space with finite dimensions. Under this restriction, the highest noise mean photon number that does not result in significant error is $\bar n_{\rm env}=4$. Plots are shown in Fig.\ \ref{fig_plot01} of the information quantities for noise mean photon number $4$ and probe mean photon number $\bar n=(0.01,0.5)$. We will now review the information quantities for each case listed in Sec.\ \ref{sec_three_cases}.

\textit{Case 1}.\ The Holevo information $\chi_q$ and Fuchs upper and lower bounds for the accessible information $A_q$ for quantum illumination with joint measurement are shown in Fig.\ \ref{fig_plot01}. The difference between the upper and lower bounds of $A_q$ is, at most, 0.7\%, implying that the true accessible information is close to the Fuchs bounds.  As evident in the plot, there is a substantial difference between the $\chi_q$ and $A_q$. \\

\textit{Case 2}.\ $\chi_c$ and $A_c$: In the previous section, we assume that a heterodyne measurement is the optimal local Gaussian measurement to make on mode B. We demonstrate in Fig.\ \ref{fig_hetopt} that this is true for a typical choice of parameters. Since a heterodyne measurement on mode B collapses mode A into a distribution of coherent states, $\chi_c$ and $A_c$ were calculated by integrating the information quantities of single coherent probe ($\chi_s$, $A_s$) as as function of energy. The computed upper and lower bounds for $A_c$ are equal to within six significant figures. \\

\textit{Case 3}.\ $\chi_s$ and $A_s$: The Holevo information $\chi_s$ is plotted in Fig.\ \ref{fig_plot01}. Fuchs lower and upper bounds for $A_s$ were calculated and are equal to within seven significant figures, and are indistinguishable in Fig.\ \ref{fig_plot01}. Unlike case 1, when using a coherent state, the Holevo and accessible information differ by a small amount, only $0.4 \%$.\\

\begin{figure}
\includegraphics[width=8cm]{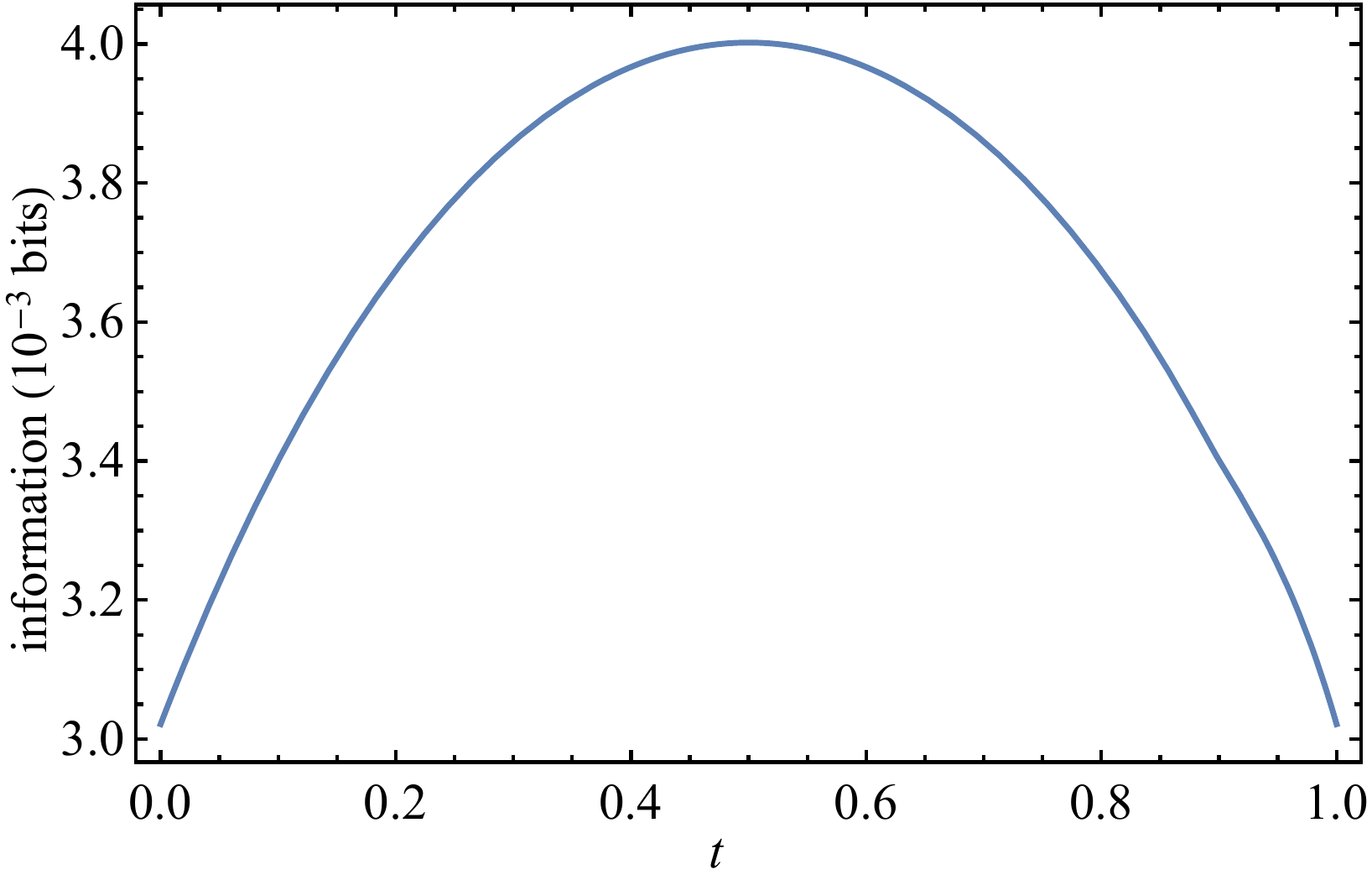}
\caption{The Holevo information obtained when an EPR state is used for illumination, but a local Gaussian measurement is first performed on mode B. The measurement consists of beam splitter with transmissivity $t$, followed by conjugate homodyne measurements on both outputs. The maximum gives $\chi_c$, which occurs up to numerical precision at $t=0.5$, which corresponds to a heterodyne measurement. Parameters are $\epsilon=0.1$, $\bar n=0.5$, and $\bar n_{\rm env}=4$. }
\label{fig_hetopt}
\end{figure}

\begin{figure*}
\includegraphics[width=15.5cm]{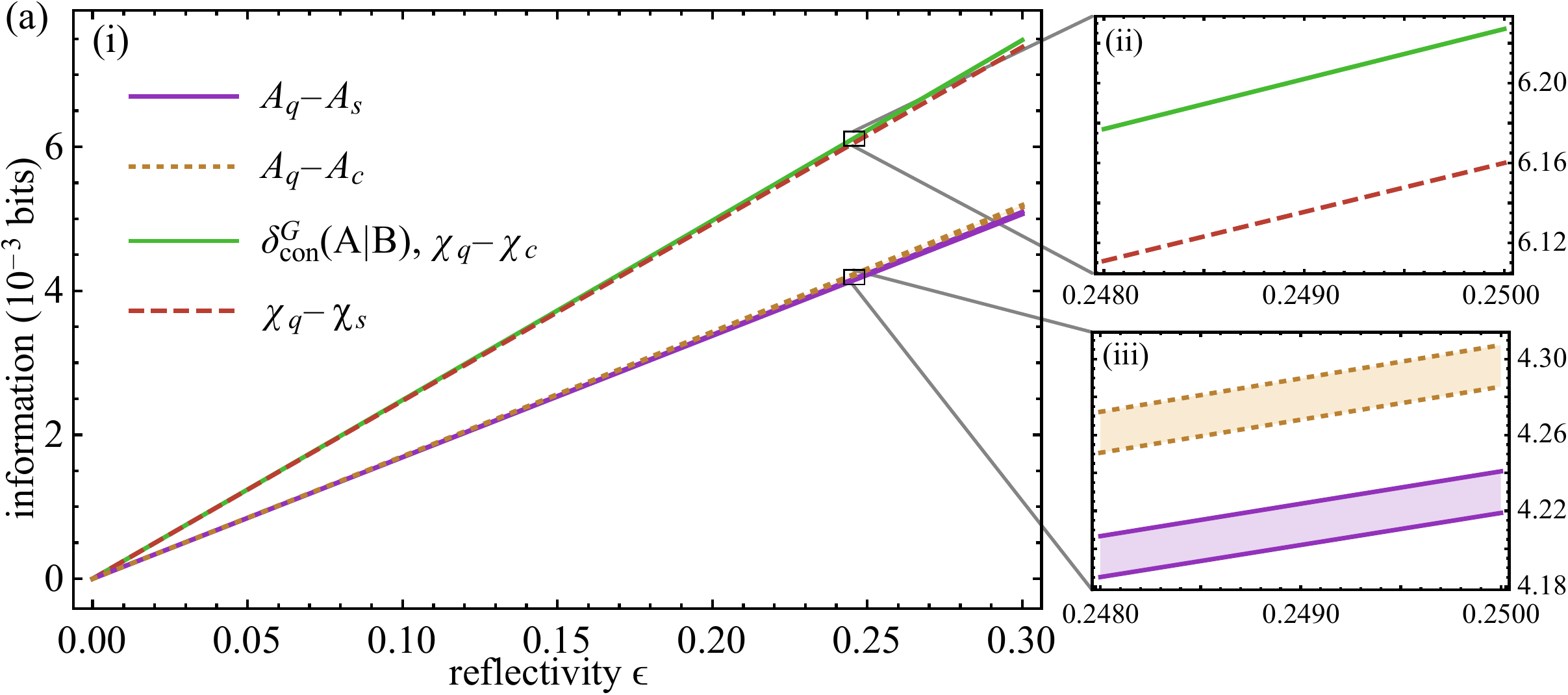}
\includegraphics[width=15.5cm]{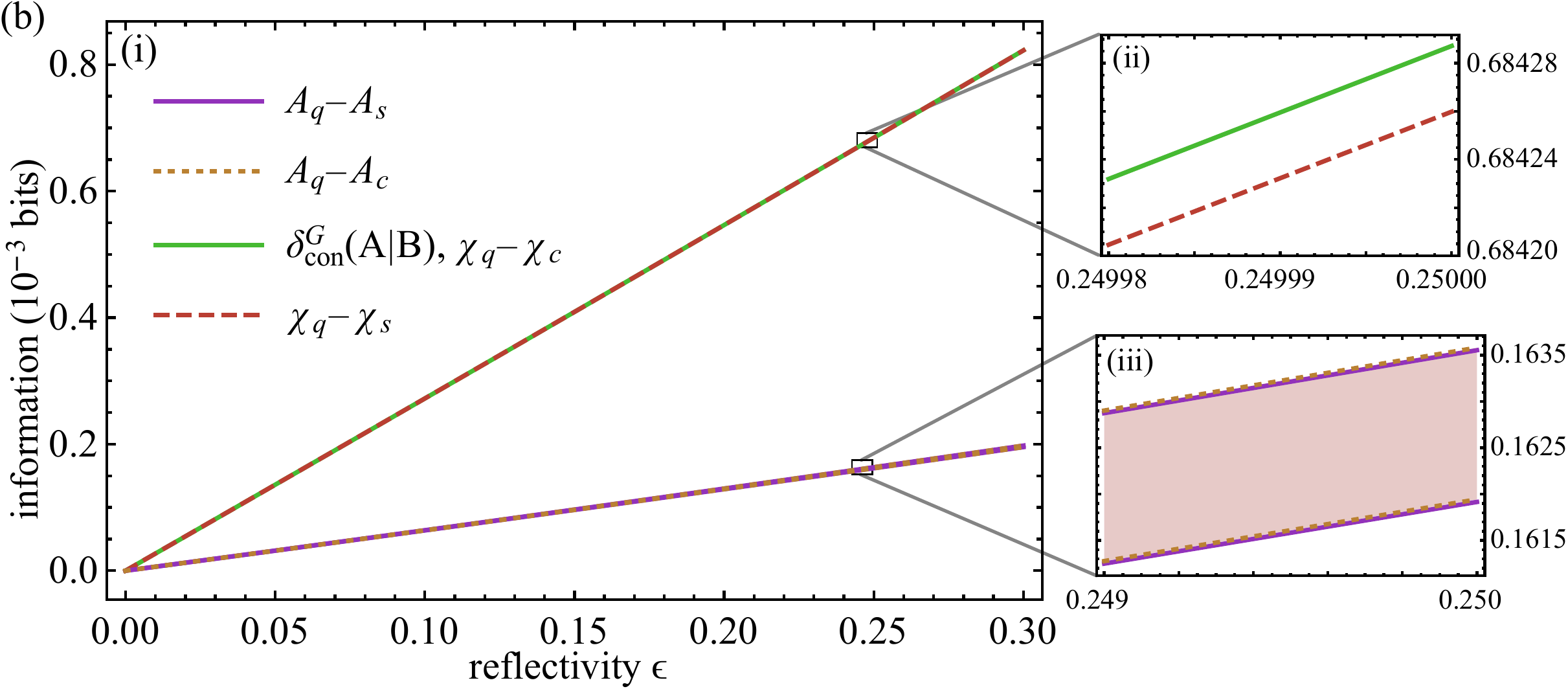}
\caption{The quantities $\chi_q-\chi_s$, $\chi_q-\chi_c$, $A_q-A_s$, and $A_q-A_c$, compared to the consumed Gaussian discord $\delta_{\rm con}^{\rm G}({\rm A}|{\rm B})$. The average photon number of the probe is (a) $0.5$ and (b) $0.01$. The mean photon number of the environment noise is $4$. Each plot has two insets showing zoomed portions. Insets (ii) show $\delta_{\rm con}^{\rm G}({\rm A}|{\rm B})$, $\chi_q-\chi_s$ and $\chi_q-\chi_c$. Insets (iii) shows upper and lower bounds of $A_q-A_s$ and $A_q-A_c$. }
\label{fig_plot02}
\end{figure*}

From Figs.\ \ref{fig_plot01}(a)(i), \ref{fig_plot01}(a)(iii), \ref{fig_plot01}(b)(i), and \ref{fig_plot01}(b)(iii), we see that $\chi_q$ is greater than $\chi_s$, and $A_q$ is greater than $A_s$, showing that quantum illumination with joint measurement does indeed have an advantage over single-mode illumination. In the communication context, Alice can communicate with Bob with a higher bit rate if Bob uses a probe entangled with an idler instead of a coherent state probe.

From Fig.\ \ref{fig_plot01}, we see that the performance of a coherent state probe is approximately equal to performance of an EPR probe when a local Gaussian measurement is performed on the mode B. However, $A_s$ is slightly higher than $A_c$ (and $\chi_s$ slightly higher than $\chi_c$) because $A_s$ is a concave function of energy (see Appendix \ref{appendix2}). By considering the ratio of $A_s$ and $A_c$, we find that their relative difference approaches zero in both the limits $\epsilon\to0$ and $\bar n\to0$. This indicates that there is no advantage to using an EPR state for illumination, over a coherent state probe, if a Gaussian measurement is first made on mode B of the EPR state. A local Gaussian measurement on mode B of an EPR state will cause mode A to collapse to a single-mode Gaussian state. Hence, this is equivalent to using a distribution of single-mode Gaussian states for the probe, which, under the masking of strong environmental noise, gives an approximately equal knowledge about a weakly reflecting object as using a single-mode coherent state probe.

\subsection{Relating Quantum Advantage to Discord Consumed}
\label{sec_results_discord}

To calculate the consumed discord $\delta_{\rm con}({\rm A}|{\rm B})$, we need to compute the discord of states $\rho^{(0)}$ and $\bar\rho$ when the entangled state $\rho_{\rm EPR}$ is used as probe and idler. $\rho^{(0)}$, the resulting state when Alice does nothing, is a Gaussian state whose discord is equal to the Gaussian discord, and additionally this discord is obtained when the measurement is a heterodyne measurement~\cite{PhysRevLett.113.140405}. The state after encoding $\bar\rho$, however, is not Gaussian, and thus the same rule does not apply. Unfortunately, calculating the discord of a general state is an NP-hard problem~\cite{1367-2630-16-3-033027}, so there is no method to calculate it efficiently. Here, we simplify the problem by restricting ourselves to Gaussian discord and calculate the consumed Gaussian discord $\delta^{\rm G}_{\rm con}({\rm A}|{\rm B})$ instead. This is just Eq.\ \eqref{eqn_denc} with the discords replaced with Gaussian discords. 

The Gaussian discord of state $\bar\rho$ was obtained by numerically optimizing Eq.\ \eqref{eq_discord} over Gaussian measurements. It was found that the optimal point occurs when the measurement is a heterodyne measurement. The two discord values $\delta^{G (0)}({\rm A}|{\rm B})$ and $\bar \delta^{\rm G}({\rm A}|{\rm B})$ are then substituted into Eq.\ \eqref{eqn_denc} to obtain the consumed Gaussian discord.

Due to the optimality of the Gaussian discord of state $\rho^{(0)}$, and the fact that Gaussian discord is an upper bound for the discord for state $\bar\rho$, the consumed Gaussian discord is a lower bound of the consumed discord, i.e., $\delta^{\rm G}_{\rm con}({\rm A}|{\rm B})\le\delta_{\rm con}({\rm A}|{\rm B})$. A plot of the $\delta^{\rm G}_{\rm con}({\rm A}|{\rm B})$ compared to the information differences is shown in Fig.\ \ref{fig_plot02}.

As discussed in Sec.\ \ref{sec_results_analytic}, since a heterodyne measurement on mode B optimizes $\delta^{(0)}({\rm A}|{\rm B})$, and numerical results show that this is the case for  $\bar \delta({\rm A}|{\rm B})$ and $\chi_c$, from Theorem \ref{theorem1}, $\delta^{(0)}({\rm A}|{\rm B})=\chi_q-\chi_c$. Numerical calculation of $\delta^{(0)}({\rm A}|{\rm B})$ and $\chi_q-\chi_c$ agree within the precision of the calculation, further verifying the theorem.

From Fig.\ \ref{fig_plot02}, we see that the difference in Holevo information between quantum illumination ($\chi_q-\chi_c$) and single-mode illumination ($\chi_q-\chi_s$) is $1.3\%$ for $\bar n=0.5$ and $0.005 \%$ for $\bar n=0.01$ when $\epsilon=0.3$. The percentage difference approaches zero when $\epsilon\to0$. Since $\delta^{\rm G}_{\rm con}({\rm A}|{\rm B})=\chi_q-\chi_c$, this leads us to the conclusion that in the limit of low reflectivity and low probe energy, $\chi_q-\chi_s$ converges to the Gaussian discord consumed.  Hence, discord encoded can suitably explain the quantum advantage of quantum illumination, if quantum illumination is viewed as a communication problem with access to devices such as quantum memory.

On the other hand, $A_q-A_s$, which quantifies the performance advantage for quantum illumination in the single-copy measurement case, is more relevant from a practical point of view since this does not require the storage of quantum states \cite{SHGSC16}. From Fig.\ \ref{fig_plot02}, we see that $\delta^{\rm G}_{\rm con}({\rm A}|{\rm B})$ is greater than $A_q-A_s$ and $A_q-A_c$. This discrepancy is mainly due to the difference between the Holevo information $\chi_q$ and the accessible information $A_q$ for the states involved in quantum illumination. Hence, measuring each illumination event separately does not fully harness the benefits offered by the discord. However, it is sufficient to provide some quantum advantage over single-mode illumination.

\subsection{Quantum advantage versus probe energy}

\begin{figure}
\includegraphics[width=8.6cm]{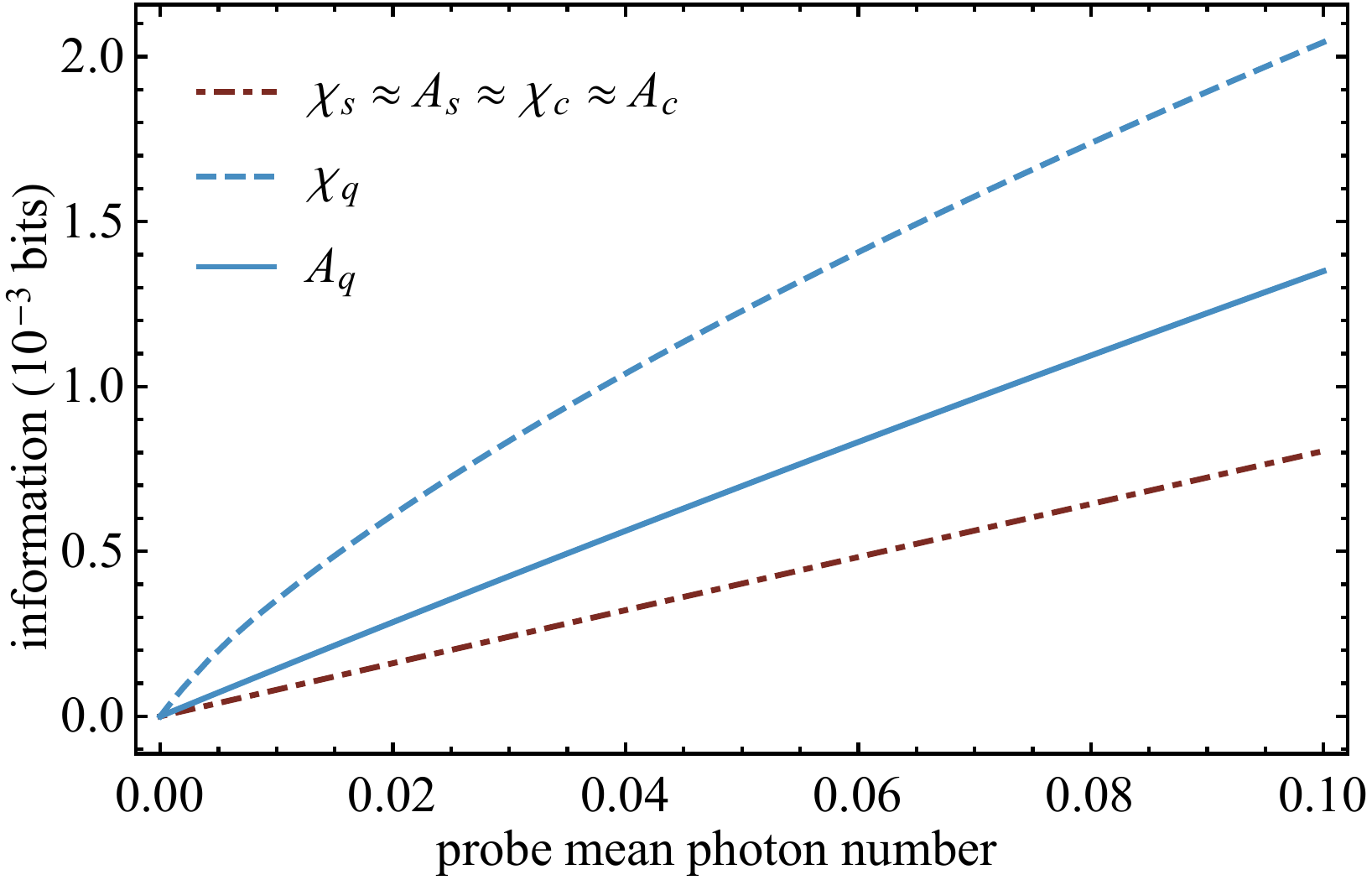}
\includegraphics[width=8.6cm]{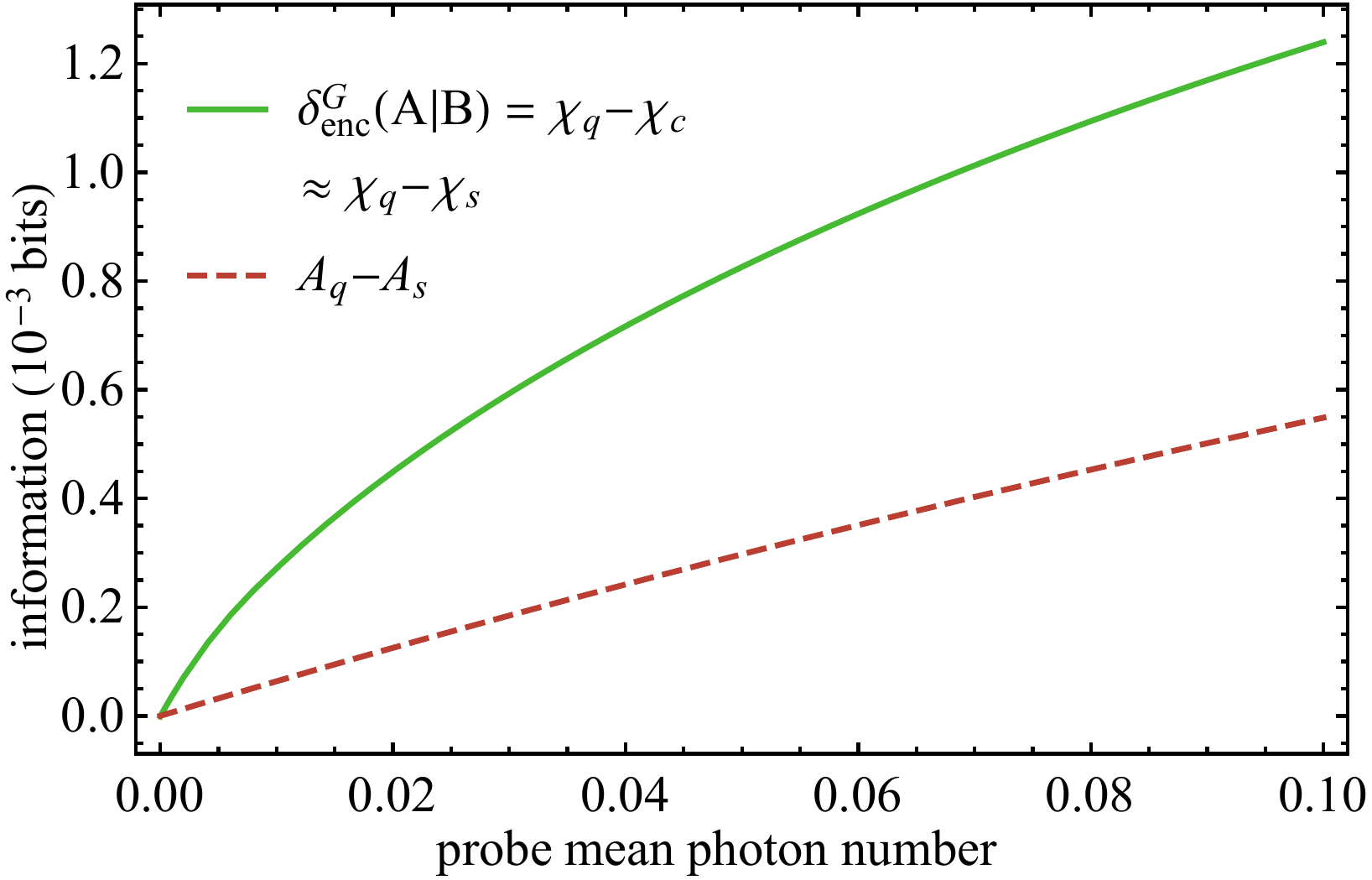}
\caption{ The accessible information and Holevo information quantities (top) and consumed Gaussian discord and quantum advantage (bottom) vs the mean energy of the probe. The environment-noise mean photon number is $4$ and object reflectivity $\epsilon=0.1$. }
\label{fig_plot_energy}
\end{figure}

There is nothing special about our choice of probe energies of 0.01 and 0.5 used in the previous sections. To demonstrate this, Fig.\ \ref{fig_plot_energy} shows the illumination performance, quantum advantage, and consumed Gaussian discord for probe mean photon numbers in the range 0 to 0.1, while the object reflectivity is kept constant at $0.1$. There is always a quantum advantage, and the consumed Gaussian discord is approximately equal to the quantum advantage in terms of Holevo information.

\subsection{Comparison to discrete variables}

\begin{table*}
\caption{Comparison of continuous-variable and discrete-variable illumination.}
\label{table1}
\begin{tabularx}{17cm}{|X|Sc|Sc|}
\hline
& {\bf DV} & {\bf CV}
\\ \hline
Environment noise &
\parbox{5cm}{
Maximally mixed \\
$\sum_{n=0}^N \frac{1}{N+1} \ket{n}\bra{n}$} &
\parbox{5cm}{
Thermal state \\
$\sum_{n=0}^\infty \frac{\bar n^n}{(\bar n +1)^{n+1}} \ket{n}\bra{n}$}
\\ \hline
Quantum illumination probe & \parbox{5cm}{Maximally entangled\\
$\sum_{n=0}^N \frac{1}{\sqrt{N+1}} \ket{n,n}$} &
\parbox{5cm}{
EPR \\
$\sqrt{1-\lambda^2}\sum_{n=0}^\infty (-\lambda)^n \ket{n,n}$}
\\ \hline
Accessible vs Holevo information &
\parbox{5cm}{
$\chi_q=A_q$ \\
no difference} &
\parbox{5cm}{
$\chi_q>A_q$ \\
big difference }
\\ \hline
Single-mode illumination &
\parbox{5cm}{
probe: any pure state $\ket{\psi}$ \\
$\chi_s=A_s$
} &
\parbox{5cm}{
probe: coherent state $\ket{\alpha}$ with $|\alpha|^2=\bar n$ \\
$\chi_s\approx A_s$
}
\\ \hline
Quantum vs single-mode illumination &
\parbox{5cm}{
$\chi_s<\chi_q$ \\
$A_s < A_q$} &
\parbox{5cm}{
$\chi_s<\chi_q$ \\
$A_s < A_q$}
\\ \hline
Single mode probe vs local measurement on idler first &
\parbox{5cm}{
$\chi_c=\chi_s$ \\
$A_c=A_s$} &
\parbox{5cm}{
$\chi_c\approx\chi_s$ \\
$A_c\approx A_s$ \\
approximation gets better with low probe energy and low reflectivity}
\\ \hline
Consumed discord vs Holevo quantum advantage &
\parbox{5cm}{
$\delta_{\rm con}=\chi_q-\chi_c$ \\
$=\chi_q-\chi_s$} &
\parbox{5cm}{
$\delta^{\rm G}_{\rm con}=\chi_q-\chi_c$ \\
$\approx\chi_q-\chi_s$}
\\ \hline
Consumed discord vs accessible info quantum advantage &
\parbox{5cm}{
$\delta_{\rm con}=A_q-A_c$ \\
$=A_q-A_s$} &
\parbox{5cm}{
$\delta^{\rm G}_{\rm con}>A_q-A_c\approx A_q-A_s$ }
\\ \hline
\end{tabularx}
\end{table*}

It is worth comparing continuous-variable (CV) illumination to discrete-variable (DV) illumination~\cite{weedbrook2016discord}. In discrete variables, the environmental noise is often described as white noise. This scenario is not realistic in continuous variables, as it corresponds to a thermal state at infinite temperature, and thus is of unbounded energy. Using a maximally mixed environment noise for DV illumination has the consequence that all pure state probes yield the same information for single-mode illumination. This is clearly not the case for any physically relevant cases of CV illumination, where a coherent state with a high energy generally performs better than a coherent state with low energy.

The probe used for quantum illumination for DV illumination is a maximally entangled state. Again, this state in CV illumination would have unbounded energy. A maximally entangled probe and idler, and a maximally mixed environment, mean that $\rho^{(0)}$ and $\rho^{(1)}$ commute in DV illumination. Hence, the Holevo information and accessible information are equal. This is not the case for CV illumination. From Fig.\ \ref{fig_plot01}, we see the differences between $A_q$ and $\chi_q$ can be significant, though deviations between $A_c$ and $\chi_c$ remain small. Quantum advantage, though, remains significant for both Holevo and accessible information.

In DV illumination, performing a local measurement on the idler first, followed by a local measurement on the probe, yields identical information as single-mode illumination. For CV illumination, this is only approximately true; these two quantities approach equality in the limit of low reflectivity and low probe energy.

Finally, in DV illumination, the consumed discord is exactly equal to the Holevo information quantum advantage and the accessible information quantum advantage. We found, for CV illumination, that this approximately holds for the Holevo information, but not for the accessible information. The differences between DV and CV illumination are summarized in Table \ref{table1}.

\section{Conclusion}
\label{sec_conclusion}

\begin{figure}[b]
\includegraphics[width=8.6cm]{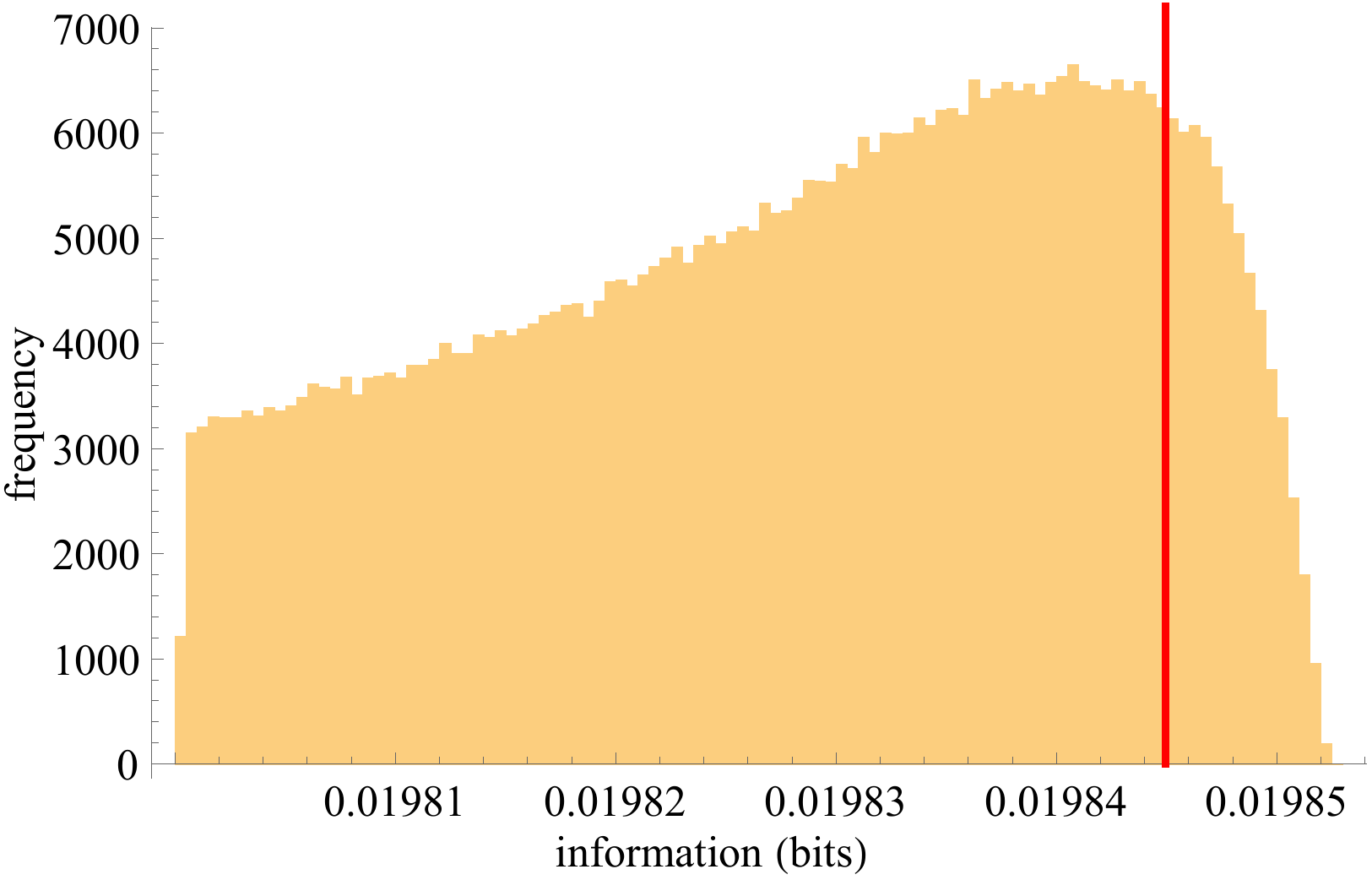}
\caption{ One million small perturbations were made on a coherent state. This is a histogram of the Holevo information when the top 500 000 states are used in single-mode illumination. The vertical line shows the Holevo information of the coherent state. Parameters are $n_{\rm env}=4$, $\epsilon=0.5$, $\bar n=0.5$. }
\label{fig_perturbed}
\end{figure}

In \cite{weedbrook2016discord}, it has been shown that quantum discord coincides exactly with quantum advantage in a DV quantum illumination. Here, we complete the picture by extending the framework to CV quantum illumination \cite{PhysRevLett.101.253601}. To this end, we numerically calculated the performance enhancement that quantum illumination has over single-mode illumination and compared it to the Gaussian discord of the system. We derived an analytic result showing that $\delta^{\rm G}_{\rm con}({\rm A}|{\rm B})=\chi_q-\chi_c$, provided condition 3 of Theorem \ref{theorem1} is met. Our main result is that the quantum advantage in terms of Holevo information matches the consumed discord in the limit of low probe energy and low object reflectivity ($\bar n\to0$ and $\epsilon\to0$). This is in agreement with the DV counterpart, which analogously assumes a maximally entropic illumination environment.

Several remarks in relation to other works are in order. In deriving our results, we have demonstrated that a joint measurement over the returning probe and idler is necessary to exploit the surviving quantum correlation to determine the non-unitary encoding. Similar to \cite{Gu:2012aa}, a coherent interaction is required to unlock the information encoded via unitary discord consumption. The discrepancy between the quantum advantage offered by Holevo information and accessible information is in concordance with recent findings, where the improvement of error probability of quantum illumination over single-mode illumination is limited to 3 dB (out of a maximum gain of 6 dB) for single-copy separate measurement in the intense white-noise limit \cite{PhysRevA.80.052310,SHGSC16}.

We note other efforts in quantifying the source of enhancement in quantum-illumination like protocols. In \cite{ragy2014quantifying}, mutual information is used to quantify the advantage offered by an entangled source over a correlated thermal source. Gaussian discriminating strength is proposed to distinguish the absence or presence of a set of unitary operations in \cite{farace2014discriminating,rigovacca2015gaussian}. The role of correlation in the improvement of channel loss detection is also established by linking discord to the performance numerically \cite{invernizzi2011optimal}. Meanwhile, several other cryptographic and metrological variants of illumination has been proposed and demonstrated recently \cite{PhysRevA.80.022320,PhysRevLett.111.010501}, in which we envisage our framework would shed light in understanding the discord's role in their quantum enhancement. 

\section*{Acknowledgements} 

\begin{figure}[b]
\includegraphics[width=8.6cm]{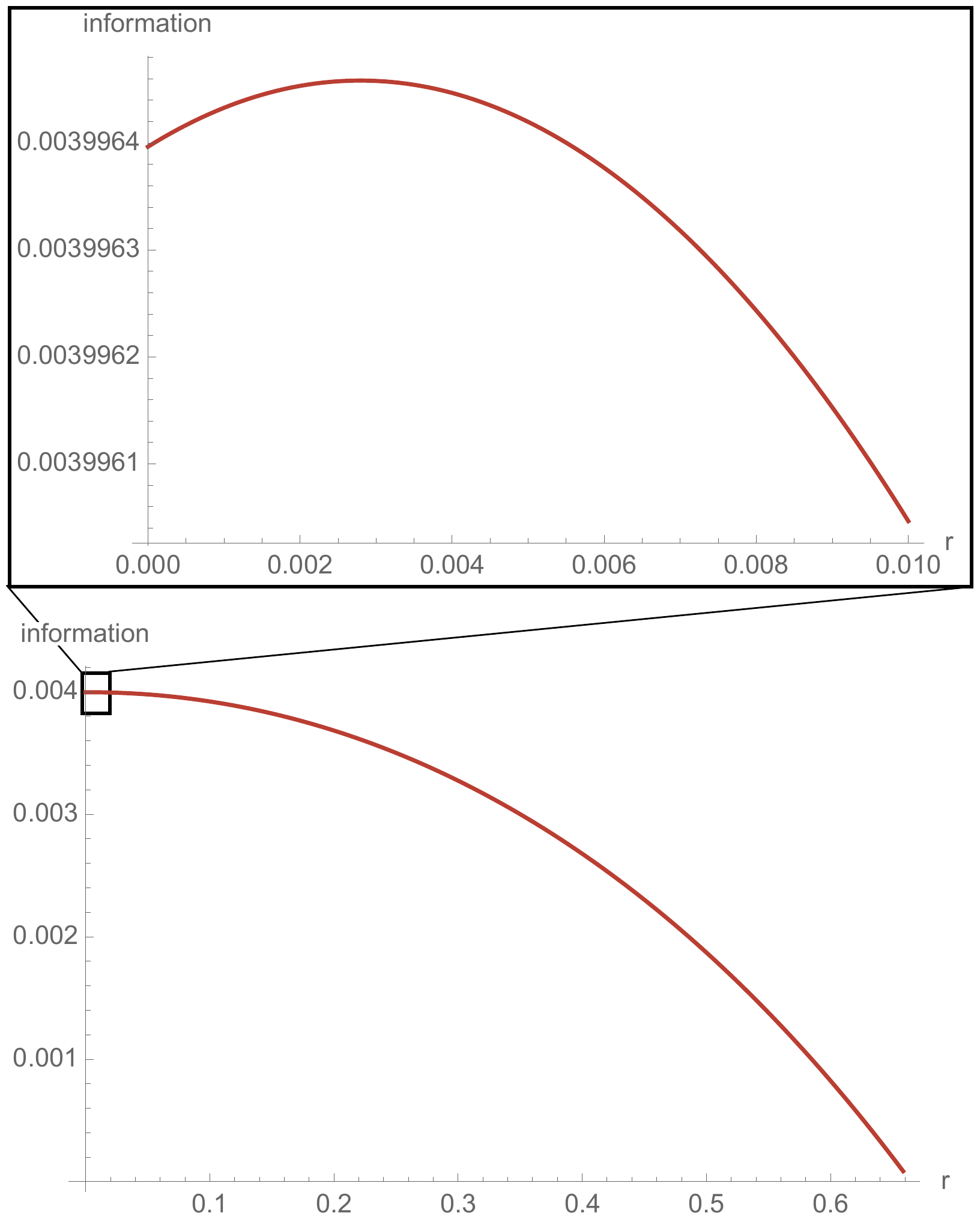}
\caption{ Accessible information when the probe is a squeezed coherent state with mean photon number $\bar n=0.5$, squeezing $r$ and displacement $\alpha=\sqrt{\bar n-\sinh^2(r)}$. The object has reflectivity $\epsilon=0.1$ and the mean photon number of the noise $\bar n_{\rm env}$ is $4$. The optimal squeezing is $r=0.00279$, which gives information $0.0015\%$ higher than with no squeezing. }
\label{fig_plotsq}
\end{figure}

We are grateful for funding from the National Research Foundation of Singapore (NRF Award No. NRF–NRFF2016–02), the John Templeton Foundation Grant No.\ 53914 {\em ``Occam's Quantum Mechanical Razor: Can Quantum theory admit the Simplest Understanding of Reality?''}, the Foundational Questions Institute, and the Australian Research Council Centre of Excellence for Quantum Computation and Communication Technology (Project No.\ CE110001027). This material is based on research supported in part by the National Research Foundation of Singapore under NRF Award No. NRF-NRFF2013-01. S.T.\ acknowledges support from the AFOSR under Grant No.\ FA2386-15-1-4082.

\begin{figure}[b]
\includegraphics[width=8.6cm]{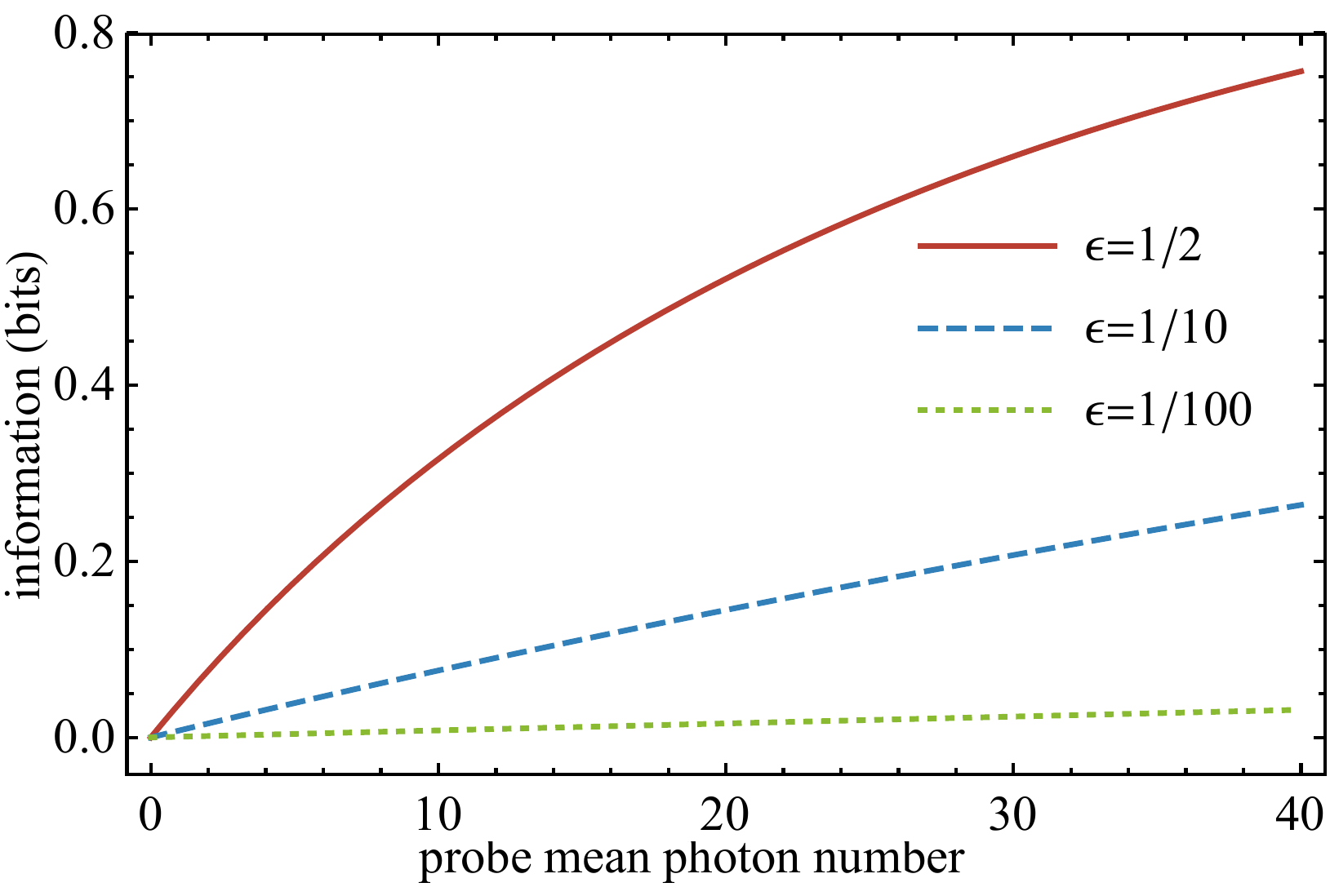}
\includegraphics[width=8.6cm]{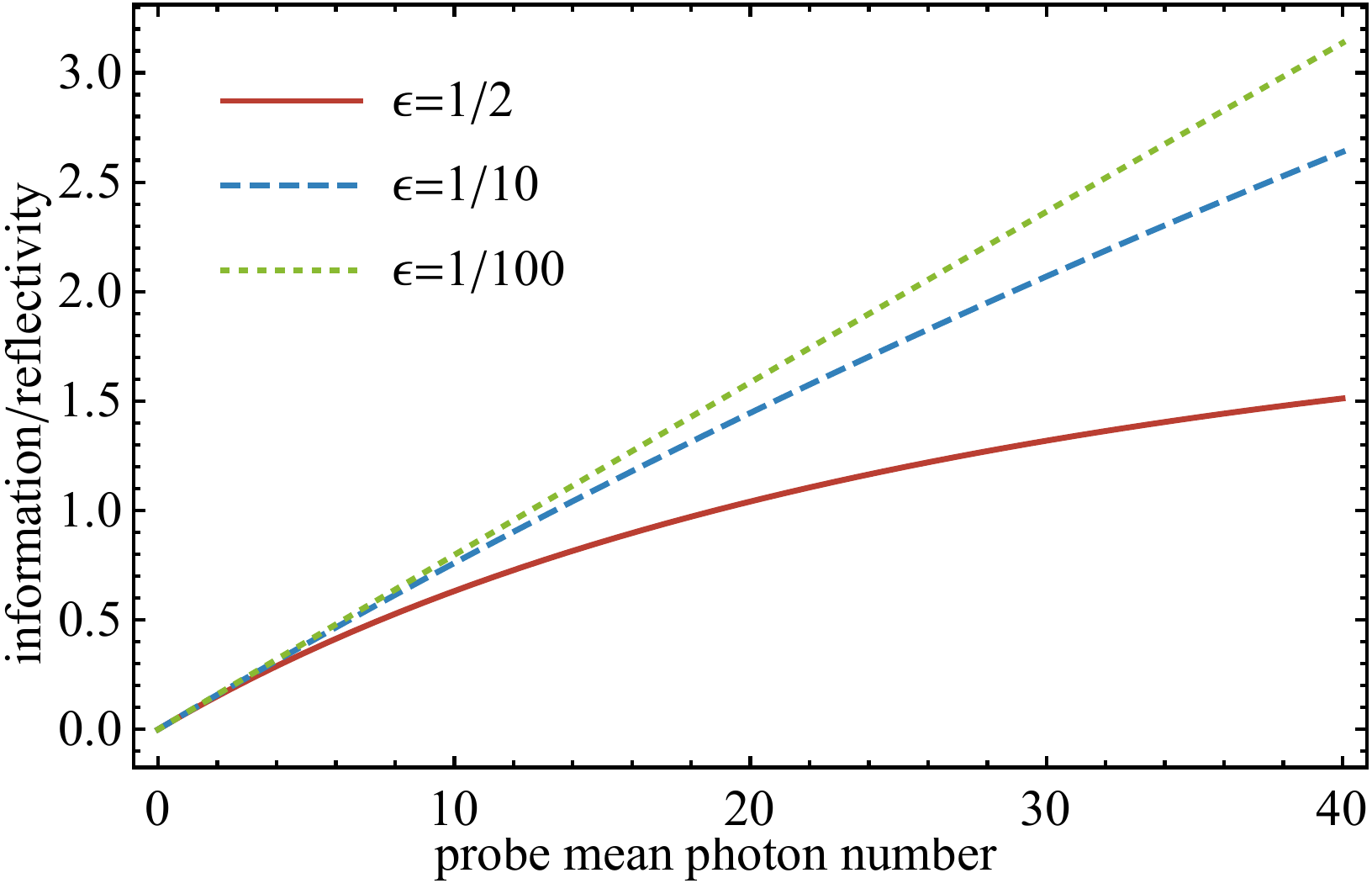}
\caption{Lower bound of $A_s$ as a function of probe average photon number $\bar n=|\alpha|^2$ when the noise mean photon number $\bar n_{\rm env}=4$, and the object has reflectivity $\epsilon$ of $1/2$, $1/10$ and $1/100$ (top). The same plot with each line scaled by $1/\epsilon$ so that linearity can be compared (bottom). }
\label{fig_plotic}
\end{figure}

\appendix

\section{Suboptimality of coherent state probe}
\label{appendix1}

A coherent state is not the optimal state to use for single-mode illumination. Small perturbations were made on a coherent state, such that the mean photon number was maintained. Figure \ref{fig_perturbed} shows a histogram of the Holevo information when the perturbed states were used in illumination. Some of the perturbed states resulted in a Holevo information greater than that achieved with the coherent state. Hence, a coherent state is not the optimal probe to use in single-mode illumination. However, we hypothesize that it is close to optimal. The problem of finding the optimal probe is too difficult to calculate, so this hypothesis is difficult to prove.

If the probe is restricted to a Gaussian state, as in Gaussian single-mode illumination, the coherent state is still not optimal. Using a squeezed coherent state with a tiny squeezing can result in increased accessible information (as can be seen in Fig.\ \ref{fig_plotsq}), but the improvement is negligible. Hence, a coherent state is approximately optimal for Gaussian single-mode illumination.

\section{Calculating $\chi_c$ and $A_c$ from integration of $\chi_s$ and $A_s$}
\label{appendix2}

From Fig.\ \ref{fig_plotic}, we see that $A_s$ is a concave function of energy. If $A_s$ were a perfect linear function of energy, $A_c$ and $A_s$ would be equal. As can be seen from Fig.\ \ref{fig_plotic}, $A_s$ as a function of energy becomes more linear as the $\epsilon$ approaches zero. Hence, this suggests that $A_c$ and $A_s$ become equal as $\epsilon$ approaches zero. The same applies to $\chi_c$ and $\chi_s$.

\bibliography{bibliography}

\end{document}